\documentclass[11pt,a4paper]{article}
\usepackage[latin1]{inputenc}
\usepackage{amsmath}
\usepackage{amsfonts}
\usepackage{amssymb}
\usepackage{graphicx}
\author{Gershon Wolansky}
\usepackage{amsmath,amssymb,amsthm,color,graphicx, fancybox}
\usepackage{tcolorbox}
\usepackage{bclogo}
\usepackage[latin1]{inputenc}
\usepackage{mdframed}
\usepackage{amsmath}
\usepackage{amsfonts}
\usepackage{amssymb}
\usepackage{wasysym}
\usepackage{amsthm}
\usepackage[english]{babel}

\usepackage{subfigure}
\usepackage{subfigmat}
\usepackage{float}

\usepackage{subfigure}
\usepackage{subfigmat}
\usepackage{float}

\newcommand{\bH}{\mathbb{H}}

\newcommand{\R}{\mathbb R}
\newcommand{\Ll}{\mathbb L}

\newcommand{\eps}{\epsilon}

\newtheorem{theorem}{Theorem}
\newtheorem{lemma}{Lemma}[section]
\newtheorem{prop}{Proposition}
\newtheorem{cor}{Corollary}[section]
\newtheorem{ex}{Problem}[section]

\newtheorem{defi}{Definition}[section]
\newtheorem{acknowledgment*}{Acknowledgment}

\newcommand{\be}{\begin{equation}}
\newcommand{\ee}{\end{equation}}

\title{\Large \textbf{ Dual formulation of  constraint solutions  for the  multi-state Choquard equation} }
\author{\textsc{Gershon Wolansky }\footnote{Department of Mathematics, Technion, Israel Inst. of Technology}}

\begin{document}
	\maketitle
	\begin{abstract}
The Choquard equation is a partial differential equation that has gained significant interest and attention in recent decades. It is a nonlinear equation that combines elements of both the Laplace and Schrödinger operators, and it arises frequently in the study of numerous physical phenomena, from condensed matter physics to nonlinear optics.

In particular, the steady states of the Choquard equation were thoroughly  investigated  using a variational functional acting on the wave functions. 

 In this article, we introduce a dual formulation for the variational functional in terms of the potential indiced by the wave function, and use it to   explore the existence of  steady states of a multi-state  version the Choquard equation in critical and sub-critical cases. 
\end{abstract}
	\section{Introduction}
	\subsection{Background}\label{backgr}
	The Choquard equation 
	$$ -\Delta \phi +   \phi -\left(\int\frac{ |\phi(y)|^2}{|x-y|} dy\right)\phi(x)=0$$
	in $\R^3$  
	was originally proposed by  Ph. Choquard, as an approximation to Hartree-Fock theory for a one component plasma. Equation of similar types also   appear to be a prototype of the so-called nonlocal problems, which arise in many situations (see, e.g \cite{Wie}) and  as a model of self-gravitating matter \cite{Pen}. 
	
	A generalized version in $\R^n$ takes the form 
	\be\label{eq1new} -\Delta \phi +\phi=\left(I_\alpha * |\phi|^p\right) |\phi|^{p-2} \phi\ee
	 where 
	 \be\label{Aalpha}I_\alpha=A(\alpha) |x|^{\alpha-n} ; \ \ \ A(\alpha):= \frac{ \Gamma\left( \frac{n-\alpha}{2}\right)}{2^\alpha\pi^{n/2} \Gamma\left(\alpha/2)\right) }\ee
	is the Rietz potential, 
	$\alpha\in (0,n)$, 
	$p\in (1,\infty)$ was considered by many authors in the last decades, using its variational structure as a critical point of the functional
	\be\label{Epa}  E_{p,\alpha}(\phi)=\frac{1}{2}\int_{\R^n}\left( |\nabla u|^2 + |\phi|^2-\frac{1}{2p}  \left(I_\alpha * |\phi|^p\right)|\phi|^p\right)\ee
	on an appropriate space. In particular, existence of solutions the case $p=2$ (and for more general singular interaction kernels) was studied by E.H. Lieb, P.L Lions and G. Menzala \cite{mez, Lion, leib}. For existence, regularity and asymptotic behavior of solutions  in the general case 
	see \cite{moroz1, moroz2} and references therein. 
	
	The  non-linear Schrödinger equation associated with $E_{p,\alpha}$  takes the form
	\be\label{nonse0}-i\partial_t\psi -\Delta\psi -a (I_\alpha*|\psi|^p)|\psi|^{p-2}\psi=0 \ . \ee
		The number $a\in \R$ is the strength of interaction. 	The case $a>0$ corresponds to the {\em attractive, gravitation-like} dynamics, and is related to  Choquard's equation.  The case $a<0$ is the repulsive, electrostatic case and is related to  the Hartree system (see, e.g. \cite{wol}). In this paper we deal with the attractive case.
		\par
	Considering an eigenmode $\psi=e^{-i\lambda t}\phi$  we get that  $\phi$ satisfy the non-linear eigenvalue problems 
		\be\label{nonsr} -\Delta \phi -a\left(I_\alpha * |\phi|^p\right) |\phi|^{p-2} \phi-\lambda\phi =0\ee
	which can be reduced to (\ref{eq1new}) by a proper scaling\footnote{Note that $\lambda<0$ is an eigenvalue below the essential spectrum of $-\Delta$} . However, the solutions of the nonlinear equation (\ref{nonse0}) preserve the $\Ll^2$ norm, so it is natural to look for stationary solutions  (\ref{nonsr}) under  a prescribed $\Ll^2$ norm (say, $\|\phi\|_2=1$). 
		It is not difficult to see that, in general, one can find a scaling $\phi \mapsto \phi_\eps(x)=\eps^{-n/2}\phi(\eps/ x)$ which preserves the $\Ll^2$ norm and transform the strength of interaction in (\ref{nonsr})   into $a=1$, making this parameter   mathematically insignificant. 
	There is, however, an exceptional case  $\alpha=n(p-1)-2$. In that case the first two terms in (\ref{nonsr}) are transformed with equal coefficients under $\Ll^2$ preserving scaling, so the size of the interaction coefficient $a$ is mathematically significant in that case.

	In the  case $p=2$ and in the presence of a prescribed, confining  potential $W$, the $\Ll^2-$ constraint version of  (\ref{nonsr}) takes the form 
	\be\label{secsc} -\Delta \phi +W \phi -a \left(\int_{\R^n} \frac{|\phi(y)|^2}{|x-y|^{n-\alpha}} dy\right) \phi-\lambda \phi=0, \ \ \ \|\phi\|_2=1 \ . \ee
	A solution of (\ref{secsc}) is given by a minimizer   of the functional 
	\be\label{EWdef} E^W_a(\phi):=\frac{1}{2} \int_{\R^n} \left(|\nabla \phi|^2 + W|\phi|^2\right)dx - \frac{a}{4} \int_{\R^n}\int_{\R^n} \frac{|\phi(x)|^2|\phi(y)|^2}{|x-y|^{n-\alpha}}dxdy \ee
	restricted to the $\Ll^2$ unit ball $\|\phi\|_2=1$.

 In {\cite{conmin} the authors studied the equation (\ref{secsc}) in the exceptional case $\alpha=n-2$, 
		for $n\geq 3$, $a>0$  and $W$ a prescribed function satisfying $\lim_{x\rightarrow\infty} W(x)=\infty$. 
	In particular, they showed the existence of a critical strength $\bar a_c >0$, depending on $n$ but independent of $W$, such   that 	  $E^W_a$  is bounded from below on the sphere $\|\phi\|_2=1$   iff $a\leq \bar a_c$. 
	Moreover, a minimizer of $E_a^W$  exists if $a<\bar a_c$, and is a solution of (\ref{secsc}) (c.f. \cite{conmin}). 
It was also shown that 
	$a_c=\|\bar\phi\|_2$, where $\bar\phi$ is the unique, positive solution   (c.f. \cite{Ma}) of the equation   of
		\be\label{barphi} -\Delta \bar\phi - \left(\int_{\R^n} \frac{|\bar\phi(y)|^2}{|x-y|^2} dy\right) \bar\phi+\bar\phi=0  \ .  \ \ee

		The object of the present  paper  is two-fold. 
		
		The first object is to extend the $\Ll^2$-constraint Choquard equation (\ref{secsc})   into a  $k-$ state system
		\be\label{Srsys} -\Delta\phi_j +W\phi_j -a \left(\sum_{i=1}^k \beta_i \int_{\R^n}\frac{|\phi_i|^2(y)}{|x-y|^{n-\alpha}} dy\right) \phi_j -\lambda_j\phi_j =0\ \  \|\phi_j\|_2=1, \ \ ; \ \ j=1\ldots k\ee
		where $(\phi_1, \ldots \phi_k)$ constitutes an orthonormal $k-$sequence in $\Ll^2(\R^n)$ and    
		\be\label{normalbeta} \beta_j>0, \ \ \ \ \sum_1^k \beta_j=1\ \ee
		are  the  {\em probabilities  of occupation} of the states $j=1\ldots k$,
		
		 In Section \ref{secmf} we introduce the time dependent Heisenberg system which leads naturally to (\ref{Srsys}), while the steady state (\ref{Srsys}) and its constraint variational formulation  are  introduced in Section \ref{steadystate}. 
		
		The second object is to 
		 introduce a dual approach to the  $\Ll^2$ constraint Choquard problem  in the case $p=2$. For the case of single state $k=1$, the dual formulation of $E^W_a$ (\ref{EWdef}) for $\alpha=2$  on the constraint $\Ll^2$ sphere  takes the form of the functional $V\mapsto {\cal H}^{W,\alpha}_a(V)$
		$$ {\cal H}^{W,2}_a(V)= \frac{a}{2}\int_{\R^n} |\nabla V|^2 + \lambda_1(V)$$
		over the {\em unconstrained}  Beppo-Levi space $V\in \dot\bH_1(\R^n)$ (c.f. section  \ref{crush}). Here the functional $\lambda_1=\lambda_1(V)$ is the leading (minimal)  eigenvalue of the Schrödinger operator 
	$-\Delta +W-aV$  on $\R^n$.
		
	 The extension of this dual formulation to the $k-$system (\ref{Srsys}) for $\alpha\in (0,2]$ is introduced in (\ref{E=H}). In case $\alpha=2$  it takes the form 
		$$ {\cal H}^{W,2}_{\beta,a}(V)= \frac{a}{2}\int_{\R^n} |\nabla V|^2 + \sum_{j=1}^k\beta_j\lambda_j(V)$$
		where $\lambda_1(V)<\lambda_2(V)\leq \ldots \lambda_k(V)$ are the leading $k$ eigenvalues of the 
		Schrödinger operator,   
while $\beta_1>\beta_2>\ldots \beta_k>0$. 

		The main result of this paper is summarized below ( Section \ref{maintheorem}):

		Using the dual variational formulation we show the existence of a minimizer of ${\cal H}^{W,\alpha}_{\beta,a}$ corresponding to a solution of  (\ref{Srsys})  in $\R^n$ for any $a>0$ where $\alpha\in (0,2]$,  $3\leq n<2+\alpha$ . In the critical cases $\alpha=2, n=4$ and $\alpha=1, n=3$ we  show the existence of a critical interaction level  $a^{(n)}_c(\beta)$ for which there is a minimizer of ${\cal H}^{W,\alpha}_{\beta,a}$ if $a<a^{(n)}_c(\beta)$ corresponding to a solution of (\ref{Srsys}),
		 while ${\cal H}^{W,\alpha}_{\beta,a}$ is unbounded from below   for any $a>a^{(n)}_c(\beta)$ for $n=3,4$.


	 \subsection{Mean-field Heisenberg system}\label{secmf}  
	 Consider the Von Neumann-Heisenberg equation
	 \be\label{1} i\frac{\partial R}{\partial t} =\left[ L^W -aV,R\right] \ , \ t\in \R\ee
	 on a Hilbert space $\bH$. Here $R$ is a density operator, namely  a bounded linear operator on $\bH$ which is self-adjoint, non-negative and of trace equal one. $L^W$ is an Hermitian  operator generating a norm preserving group $e^{itL^W}$ on $\bH$ and $V$ is a non-linear operator. 
	 
	 In the context of mean-field system 	we consider  $(\bH, \left<\cdot, \cdot\right>)$ to be the Hilbert space $\Ll^2(\R^n)$ where $\left<\phi,\psi\right>:=\int_{\R^n} \phi\bar\psi$ the canonical inner product. A density operator can be represented by a kernel $K_R$ acting on $\phi\in \bH$ via 
	 $ R(\phi)=\int_{\R^n}K_R(x,y)\phi(y)dy$
	  and 
	  $Tr(R)(x):= K_R(x,x) $. In these terms we define $V(R)$ as the  operator acting on $\phi\in\bH$ by multiplication with 
	  
	    \be\label{2} V(R):=  I_\alpha *Tr(R )\ \ \ee

	 Since $L^W-aV$ is hermitian for any prescribed potential $V$, all observables  along the orbit $t\mapsto R(\cdot, t)$ are unitary equivalent:
	 \be\label{ue}R(\cdot,t)= \exp\left(-i\int_0^t (L^W-aV(\cdot, s))ds\right) R(\cdot, 0) \exp\left(i\int_0^t (L^W-aV(\cdot, s))ds\right) \ .  \ee
	  We restrict ourselves  to a class of observables of a {\em finite rank} $k\in\mathbb{N}$. 
	 Hence 
	 the kernel of $R$ can be represented as 
	 \be\label{Rdef} R(x,y,t)= \sum_1^k \beta_j\psi_j(x,t)\bar\psi_j(y,t)\ee
	 where $\beta_j>0$ are the eigenvalues of $R$, which are constant in time, and $\psi_j(\cdot,t)\in \bH$ constitute  an orthonormal sequence for any $t\in\R$. Under this representation (\ref{1}) takes the form 
	 \be\label{srj} i\frac{\partial \psi_j}{\partial t} = (L^W -aV)\psi_j\ \ ,  j=1,2, \ldots k \ee
	 
	 The eigenvalues  $\beta_j\in [0,1]$ are interpreted as the probability of occupation of the $j-$ level     satisfying $\sum_{j=1}^k \beta_j=1$.  For any $t\in\R$, the trace of $R$ conditioned on $x\in\R^n$ is 
	 \be\label{TrR}Tr(R)(x,t)= \sum_{j=1}^k \beta_j|\psi_j(x,t)|^2 \ , \ee
	 and the potential $V$ is determined in terms of the solution $R$ by (\ref{2})
	 $$ V= \sum_{j=1}^k \beta_j I_\alpha* |\psi_j|^2 \ . $$


	 Consider now the Hamiltonian 
	 $${\cal E}^{(\alpha)}_{\beta,a}(\vec{\psi}) := \frac{1}{2}\sum_{j=1}^k\beta_j \left[\left<L^W\psi_j, \psi_j\right> - \frac{a}{2}\sum_{i=1}^k \beta_i \left<|\psi_j|^2,  I_\alpha*|\psi_i|^2\right>\right] $$
	acting on $k-$ orthonormal  frames $\vec{\psi}=(\psi_1, \ldots \psi_k)$.  The system (\ref{1}) (equivalently (\ref{srj}) ) is, in fact, an Hamiltonian system in the canonical variables $\{\psi_i, \bar{\psi}_j\}$:
	 \be\label{psisubt} i\frac{\partial \psi_j}{\partial t} = -\frac{1}{\beta_j} \delta_{\bar{\psi}_j}{\cal E}^{(\alpha)}_{\beta,a} ; \ \ \  i\frac{\partial \bar\psi_j}{\partial t} = \frac{1}{\beta_j} \delta_{\psi_j}{\cal E}^{(\alpha)}_{\beta,a}\ . \ee
	 In particular, ${\cal E}^{(\alpha)}_{\beta,a}$ is constant along the solution of (\ref{1}). 
	 
	 	\subsection{Steady states}\label{steadystate}
	 	 The steady states of this system are given by $\psi_j=e^{-i\lambda_j}\phi_j$ where $\{\phi_j\}$ is an orthonormal sequence corresponding to  eigenvalues $\lambda_j$ of the operator $L^W-aV$, satisfying 
	 	\be\label{ss}L^W\phi_j -a I_\alpha*\left(\sum_{i=1}^k \beta_i|\phi_i|^2\right)\phi_j -\lambda_j\phi_j =0\ .  \ee
	 	\begin{defi} \label{def1.1}
	 		$$\bH^1:=\{ \phi\in \Ll^2(\R^n), \ ; \ \nabla\phi\in \Ll^2(\R^n); \   \|\phi\|_2=1, \ \ \int_{\R^n} W|\phi|^2<\infty \ \}$$ 
	 		$$ \oplus^k\bH^1:= \left\{ \vec\phi=(\phi_1, \ldots \phi_k), ; \phi_j\in\bH^1 ; \ \left<\phi_j, \phi_i\right>=\delta_i^j , \ \ i,j\in\{1, \ldots k\}\right\} \ . $$
	 		$ \left<\left< \phi,\phi\right>\right>_W$ is the quadratic form on $\bH^1$ defined by the completion of $\left<L^W\phi,\phi\right>$:
	 		$$\left<\left< \phi,\phi\right>\right>_W:=\int_{\R^n} |\nabla\phi|^2 +W|\phi|^2 \ . $$
	 		Let
	 		$${\cal E}^{(\alpha)}_{\beta,a}(\vec\phi):= \frac{1}{2} \sum_1^k\beta_j\left[  \left<\left< \phi_j,\phi_j\right>\right>_W -\frac{a }{2}
	 		\sum_{i=1}^k\beta_i\left<|\phi_j|^2, I_\alpha* |\phi_i|^2\right>\right]\  $$
	 		is defined over $\oplus^k\bH^1$ (c.f. Corollary \ref{cor2.1} below). 
	 	\end{defi}

	 We formally obtain from (\ref{psisubt}) that  the steady states (\ref{ss}) are critical points of ${\cal E}^{(\alpha)}_{\beta,a}$ subject to the orthogonality constraints. 
\begin{prop}\label{propcritical}
	Suppose $\beta_j\not=\beta_i$ for any $1\leq i\not= j\leq k$. Then any critical point of ${\cal E}^{(\alpha)}_{\beta,a}$ restricted to orthonormal frames  $\overrightarrow{\phi}=(\bar\phi_1\ldots \bar\phi_k)$   is composed of $k$ normalized  eigenstates of the operator $L^W-a\bar V$ where $\bar V=I_\alpha *\left( \sum_1^k \beta_j|\bar\phi_j|^2\right)$. 
\end{prop}
For the proof of Proposition \ref{propcritical} see the beginning of section \ \ref{proofs}.

From now on we assume 
\be\label{betadef} \beta_1>\beta_2>\ldots \ > \beta_k> 0  \ . \ee

	\paragraph{Formulation of the problem}:
	Consider the multi-state Choquard system satisfying the equivalent of  (\ref{ss}): 
	\be\label{firsteq} (L^W -aV)\phi_j -\lambda_j\phi_j=0\ \ ; j=1, \ldots \ k \ee
	on $\R^n$. Here:
	\begin{description}
		\item[i)] $L^W=-\Delta+W$, $\Delta:=\sum_{i=1}^n \partial^2_{x_i}$ is the Laplacian on $\R^n$ and 
		\be\label{Winf}  W\in \Ll^\infty_{loc}(\R^n),  \ \ \ \lim_{|x|\rightarrow\infty} W(x)=\infty, \ \ \inf_{x\in\R^n}W(x)=W(0)=0 \  \ee
		\item[ii)]  $\vec\phi=(\phi_1, \ldots \phi_k)\in\oplus^k\bH^1$ are normalized  eigenfunctions of $L^W-aV$ and 
	 $\lambda_j \in\R$  are the corresponding  eigenvalues. 
		\item[iii)] 
			\be \label{rie}  V=\sum_{i=1}^k\beta_iI_\alpha *|\phi_i|^2   \ee
			where 
	$\beta_j$ are the {\em probabilities of occupation} of the states $j$, thus $\beta_j>0$ and $\sum_1^k\beta_j=1$. 
	\item[iv)] $a>0$.  

	\end{description}


\subsection{A crush review on Rietz kernels and its dual}\label{crush}


	Let us recall some definitions and theorems  we use later (for more details see \cite{Ten}): 
	
	For $V_1, V_2\in C^\infty_0(\R^n)$ and $\alpha\in (0,n)$, consider 
the quadratic form 
$$\left<V_1, V_2\right>_{\alpha/2}:=
A(-\alpha)\int_{\R^n}
\int_{\R^n} 
\frac{(V_1(x)-V_1(y)) (V_2(x)-V_2(y))}{|x-y|^{n+\alpha}}dxdy$$
where the constant $A(-\alpha)$ is defined as in (\ref{Aalpha}). 
If $\alpha=2$
$$\left<V_1, V_2\right>_{(1)}:=\int_{\R^n}\nabla V_1\cdot\nabla V_2 dx \ . $$
The closure of $C_0^\infty(\R^n)$ with respect to the norm induced by the inner product $\left<\cdot, \cdot\right>_{\alpha/2}$ is denoted by $\dot\bH^{\alpha/2}$.   We denote the associated norm by $\||\cdot\||_{\alpha/2}$. \footnote{Note that $\dot\bH^{\alpha/2}(\R^n)$ does not contain $\Ll^2(\R^n)$. In case $\alpha=2$ it is sometimes called Beppo-Levi space.} Recall that $\dot\bH^{\alpha/2}$ is a Hilbert space so, in particular, is weakly locally compact. 
\begin{lemma}\label{critsob}\cite{hitch}
	For $\alpha\in (0,2]$, $n>2$,  the space $\dot{\bH}^{\alpha/2}$ is continuously embedded in $\Ll^{2n/(n-\alpha)}(\R^n)$, so there exists $S=S_{n,\alpha}>0$ such that
	$$ \|V\|_{2n/(n-\alpha)} \leq S_{n,\alpha}\|V\||_{\alpha/2} \ . $$
\end{lemma}
 The fractional Laplacian $(-\Delta)^{\alpha/2}$, $0<\alpha<2$ is defined as a distribution by
 $$\left<V, \phi\right>_{\alpha/2}=\left<(-\Delta)^{\alpha/2} V, \phi\right>    \  \forall \ \ \phi\in C_0^\infty(\R^n) \ . $$
and the pointwise definition of the fractional Laplacian for $0<\alpha<2$ is given in terms of the singular integral 
$$(- \Delta)^{\alpha/2} V(x)= A(-\alpha) \int_{\R^n} \frac{V(x+y)-V(x)}{|y|^{n+\alpha} }dy\ . $$
For $\alpha=2$, the above definition is reduced to the classical, local Laplacian $-\Delta=\sum_{j=1}^n \partial^2_{x_j}$. 

The Rietz potential $I_\alpha$  is defined as a distribution via the quadratic form  induced by the dual of the $\left<\cdot, \cdot\right>_{\alpha/2}$ inner product:
\be\label{Ivar} \frac{1}{2}\left<I_\alpha*\rho,\rho\right>:=  
\sup_{V\in C_0^\infty(\R^n)} \left<\rho, V\right>-\frac{1}{2}\left<V, V\right>_{\alpha/2} \ . \ee
The Euler-Lagrange equation corresponding to the right hand side of (\ref{Ivar}) takes the form
\be\label{ELalpha} (-\Delta)^{\alpha/2} V=\rho \ . \ee
 In particular, $I_\alpha\equiv (-\Delta)^{-\alpha/2}$ corresponds to the right inverse of the fractional Laplacian 
\be\label{IinvD} I_\alpha*(-\Delta)^{\alpha/2}V = V\ . \ee   The  pointwise representation of the kernel $I_\alpha$  is given by  (\ref{Aalpha}).  Moreover
\begin{lemma}\cite{Rietz}\label{lemapqrietz}
	For any $0<\alpha<n$, the Rietz potential is a bounded operator  from  $\Ll^p(\R^n)$ to $\Ll^q(\R^n)$ iff $1<p<n/\alpha$ and $1/q=1/p-\alpha/n$. 
\end{lemma}
Our main results, described below, concern the Choquard problem on $\R^n$. However, in order to overcome problems of lack of compactness, we shell need to introduce a version of this problem in bounded domain $\Omega\subset \R^n$. In order to handle this we need to define the Green function corresponding to the fractional Laplacian $(-\Delta)^\alpha$ in a bounded domain under homogeneous Dirichlet condition.   This is the motivation to define the {\em local Rietz potential}  $I_\alpha^\Omega$ 
on $\Ll^p(\Omega)$ by
\be\label{defIomega}  \frac{1}{2}\left<I^\Omega_\alpha(\rho),\rho\right>:=  
\sup_{\phi\in C_0^\infty(\Omega)} \left<\rho, \phi\right>-\frac{1}{2}\left<\phi, \phi\right>_{\alpha/2} \ . \ee
In the case $\alpha=2$ this definition induces the Green function of the Dirichlet problem $I_2^\Omega\equiv (-\Delta_\Omega)^{-1}$, that is, 
$V(x)= \int_\Omega I_2^\Omega(x,y)\rho(y)dy$ is the solution of the Poisson problem 
\be\label{deltav=rho} \Delta V+\rho=0 \ \ \ x\in \Omega; \ \ V=0 \ \ \text{on} \ \ \partial\Omega \ . \ee
Not much is known\footnote{But see Section \ref{FR}-c.}  on the Green function $I_\alpha^\Omega$ for  $\alpha<2$. In case $\alpha=2$ the maximum principle implies immediately that  for any  $x,y\in\Omega$ the inequality $I_2(x-y)\geq I_2^\Omega(x,y)$ holds, and that $I_2^\Omega(x,y)=0$ if $x\in\Omega, y\in\partial\Omega$.  In the general case we  obtain from (\ref{defIomega}) :
\begin{lemma}\label{lema1.3} For any $0<\alpha<2$, 
	 $\Omega_2\supset\Omega_1$  and $supp(\rho)\subset\Omega_1$ then $I_\alpha^{\Omega_1}(\rho )
	\leq I_\alpha^{\Omega_2}(\rho)\leq I_\alpha*\rho $. 
	
	In addition: Let  $\Omega_j\subset \R^n$, $\Omega_j\rightarrow \R^n$ is a monotone sequence of domains in $\R^n$.  If $\rho_j$ converges to $\rho$ in  $\Ll^p(\R^n)$, $p\in (1, n/\alpha)$,     and  $\rho_j$ are supported in  $\Omega_j$ then 
	$$ \lim_{j\rightarrow\infty}I^{\Omega_j}_\alpha(\rho_j)=I_\alpha *\rho  \ \text {in} \ \ \Ll^{\frac{pn}{n-p\alpha}}(\R^n)\  .$$

\end{lemma}

%
%

\subsection{Spectrum of the Schrödinger  operator}
One of the most celebrated results on the discreteness of spectrum for the Schrödinger operator $-\Delta+W$  in $\Ll^2(\R^n)$  with a locally integrable potential is a result of K. Friedrichs \cite{Frid}  which  ensures the discreteness of spectrum if the potential $W$  grows at infinity at arbitrary rate. This and (\ref{Winf})  implies, in particular
\begin{prop}\label{firstprop}
Let $V\in C_0^\infty (\R^n)$ and $W$ satisfies (\ref{Winf}). 
  Then  the spectrum of the operator $L^W-aV$
 is composed of an infinite set of eigenvalues  $\lambda_j\rightarrow \infty$ and the corresponding  normalized eigenfunctions  $\phi_j $ 
 constitute a complete orthonormal base of  $\Ll^2(\R^n)$. 
\end{prop}
For the proof of Proposition \ref{secondprop} see Lemma \ref{lemma3.2}    
in  Section \ref{proofs}.  
\begin{prop}\label{secondprop}
	If $3\leq n\leq 5$, $0<\alpha\leq 2$,   then Proposition \ref{firstprop} can be extended for $V\in \dot\bH^{\alpha/2}(\R^n)$, and $V \mapsto\lambda_j(V)$ is a continuous functional  in the $\||\cdot\||_{\alpha/2}$ norm. 
\end{prop}
Proposition \ref{secondprop} allows us to define the dual functional on $\dot\bH^{\alpha/2}(\R^n)$:
	\be\label{E=H} {\cal H}^{W,\alpha}_{\beta,a}(V):= \frac{a}{2}\left<V,V\right>_{\alpha/2}
	+ \sum_{j=1}^k \beta_j\lambda_j(V)
\ee
where $\lambda_1(V)<\lambda_2(V)\leq \lambda_3(V) \ldots \leq \lambda_k(V)$ are the lowest   $k$ eigenvalues of  the operator $L^W-aV$.  

\subsection{Main Theorem} \label{maintheorem}

 {\em
For any $V\in \dot\bH_{\alpha/2}(\R^n)$ let $\phi_j(V)$ be a normalized  eigenstate corresponding to $\lambda_j(V)$
 \footnote{If $\lambda_j(V)$ is degenerate, so $\lambda_{j-1}(V)>  \lambda_j(V)= \ldots 
	 =\lambda_{j+l}(V) > \lambda_{j+l+1}(V)$,  then $\{\phi_j(V), \ldots \phi_{j+l}(V)\}$ is any orthonormal base of the eigenspace of $\lambda_j$. 
	}
	\begin{description}
 
	\item{[i]} $V$  is a minimizer of ${\cal H}^{W,\alpha}_{\beta,a}$  on $\dot\bH_{\alpha/2}(\R^n)$ if and only if $(\phi_1(V), \ldots \phi_k(V)) $ is a minimizer of ${\cal E}^{(\alpha)}_{\beta,a}$ on $\oplus^k\bH^1$.  If this is the case then 	$\{\lambda_j(V),  \phi_j(V)\}_{1\leq j\leq k}$  is a solution of (\ref{firsteq}) .

	\item{[ii]} If $\alpha\in(0,2]$  $3\leq n<2+\alpha$ then 
	 ${\cal H}^{W,\alpha}_{\beta,a}$ is bounded from below on $\dot\bH_{\alpha/2}(\R^n)$ for any $\beta$ satisfying (\ref{betadef}) and any $a>0$, and there is a minimizer of ${\cal H}^{W,\alpha}_{\beta,a}$ in $\dot\bH_{\alpha/2}(\R^n)$. 
	\item{[iii]} If $\alpha=1, n=3$ or $\alpha=2, n=4$ then there exists $a_c(\beta)>0$ such that ${\cal H}^{W,\alpha}_{\beta,a}$ is bounded from below  on $\dot\bH_{\alpha/2}(\R^n)$ if $a<a_c(\beta)$ and unbounded from below if $a>a_c(\beta)$. 
If $a<a_c(\beta)$ there exists a minimizer of ${\cal H}^{W,\alpha}_{\beta,a}$ in $\dot\bH_{\alpha/2}(\R^n)$. 
\end{description}


}
 \section{Proofs}\label{proofs}
We start by proving Proposition \ref{propcritical}


\begin{proof}
	Let $\gamma_{i,j}$ be the Lagrange multiplier for the constraints $\left<\phi_i,\phi_j\right>=\delta_{i,j}$. Then $\overrightarrow{\phi}$ is an {\em unconstraint} critical point of $${\cal E}^{(\alpha)}_{\beta,a}(\vec{\phi}) + \sum_1^k\gamma_{j,j}\|\phi_j\|_{\bH^1}^2 + \sum_{i\not= j} \gamma_{i,j}\left<\phi_i, \phi_j\right> \ . $$ 
	This implies $$ 
	\frac{\delta {\cal E}^{(\alpha)}_{\beta,a}}{\delta\bar\phi_j }+ 2\gamma_{j,j}\bar\phi_j + \sum_{i\not= j }\gamma_{i,j} \bar\phi_i=
	\beta_j (L^W-a\bar V)\bar\phi_j + 2\gamma_{j,j}\bar\phi_j + \sum_{i\not= j }\gamma_{i,j} \bar\phi_i=
	0 \ . $$
	In particular, $Sp(\bar\phi_1\ldots\bar\phi_k)$ is an invariant subspace of $L^W-a\bar V$. 
	Since $(\bar\phi_1\ldots \bar\phi_k)$ is an orthonormal sequence we get, after multiplying the above line by $\bar\phi_i$ and taking the inner product: 
	$$ \left<\beta_j\left[L^W-a\bar V\right]\bar\phi_j, \bar\phi_i\right> +\gamma_{i,j}=0$$
	for any $i\not= j$. switching $i$ with $j$ and taking into consideration that $L+\bar V$ is self-adjoint, we also get
	$$ \left<\beta_i\left[L^W-a\bar V\right]\bar\phi_j, \bar\phi_i\right> +\gamma_{i,j}=0$$
	Subtracting the two inequalities we obtain $\left<(\beta_j-\beta_i)\left(L^W-a\bar V\right)\bar\phi_j, \bar\phi_i\right>=0$ thus  $\left<\left(L^W-a\bar V\right)\bar\phi_j, \bar\phi_i\right>=0$ for any $i\not= j$. 
	Since  $Sp(\bar\phi_1\ldots\bar\phi_k)$ is an invariant subspace of $L^W-a\bar V$,  this implies that 
	$\bar\phi_j$ is are eigenstates of $L^W-a\bar V$. 
\end{proof}

The first part of the  following Lemma follows from a compactness embedding Theorem (c.f.  Theorem XIII.67 in  \cite{BS}). The second part  from Sobolev and HLS inequalities (see, e.g. \cite{moroz2}, sec. 3.1.1)

\begin{lemma} \label{lemma2.1}
For any $n\geq 3$, 	$\bH^1$ is compactly embedded in $\Ll^r$ for $2<r<\frac{2n}{n-2}$. 
	If $\frac{n-2}{n+\alpha} \leq  \frac{1}{2} \leq \frac{n}{n+\alpha}$ and $\phi\in\bH^1$ then $|\phi|^2\in \Ll^{2n/(n+\alpha)}(\R^n)\cap \Ll^1(\R^n)$ and 
	$$\int_{\R^n} (I_\alpha*|\phi|^2)|\phi|^2\leq C_{n,\alpha} \left(\int|\phi|^{4n/(n+\alpha)}\right)^{1+\alpha/n}$$
	\end{lemma}
	In particular we obtain:
	\begin{cor}\label{cor2.1}
	If $\max(0, n-4)\leq\alpha\leq n$ then 	the functional ${\cal E}^{(\alpha)}_{\beta,a}$ is defined on $\oplus^k\bH^1$. 
	\end{cor}
	Let  $\alpha\in(0,2]$, $\vec\phi\in \oplus^k \bH^1$ and $V\in C_0^\infty(\R^n)$. 
Define
\be\label{Hbetadef}{\bf  H}^{(\alpha)}_\beta(\vec\phi, V)=  \sum_1^k\beta_j \left<L^W\phi_j, \phi_j\right> +a\left[ \left<V, V\right>_{\alpha/2} -\left<V, \sum\beta_j|\phi_j|^2\right> \right]\  . 
\ee
By (\ref{Ivar}) we get (c.f Definition \ref{def1.1})
$$\inf_{V\in C_0^\infty(\R^n)} {\bf H}^{(\alpha)}_\beta(\vec\phi,V)=2 {\cal E}^{(\alpha)}_{\beta,a}(\vec\phi) \ . $$
Thus
$$\inf_{V\in C_0^\infty} \inf_{\vec\phi\in \oplus^k\bH^1}
{\bf H}^{(\alpha)}_\beta(\vec\phi,V)\equiv \inf_{\vec\phi\in \oplus^k\bH^1}\inf_{V\in C_0^\infty} 
{\bf H}^{(\alpha)}_\beta(\vec\phi,V)= \inf_{\vec\phi\in \oplus^k\bH^1}{\cal E}^{(\alpha)}_{\beta,a}(\vec\phi) \ . $$
Let
$$ {\cal H}^{W,\alpha}_{\beta,a}(V)=\inf_{\vec\phi\in \oplus^k\bH^1}
{\bf H}^{(\alpha)}_\beta(\vec\phi,V) \ . $$
From (\ref{Ivar} , \ref{Hbetadef}) we observe that 
$${\cal H}^{W,\alpha}_{\beta,a}(V)=\frac{a}{2}\left< V,V\right>_{\alpha/2}
+ \inf_{\vec{\phi}\in \oplus^k\bH^1}\sum_{j=1}^k \beta_j \left<(L^W-aV)\phi_j, \phi_j\right> \ . $$
Let 
	\be\label{minlin} \inf_{\vec{\phi}\in \oplus^k\bH^1}\sum_{j=1}^k \beta_j \left<(L^W-aV)\phi_j, \phi_j\right> := G_{\beta, a}(V)\ . \ee
	As the infimum over linear functionals,   $V\mapsto G_{\beta,a}(V)$ is a concave  functional, so
\be\label{Hdef} {\cal H}^{W,\alpha}_{\beta,a}(V)=\frac{a}{2}\left< V,V\right>_{\alpha/2}+G_{\beta, a}(V)\  \ee
is the sum of convex and concave functionals. 
In the case $k=1$ ($\vec\beta= \beta_1=1$) we observe, by the Rayleigh-Ritz principle $$G_{1,a}(V)= \inf_{\|\phi\|=1}\left<(L^W-aV)\phi,\phi\right>= \lambda_1(V)$$
 and the supremum is obtained at the normalized ground state $\bar\phi_1$ satisfying $(L^W-aV-\lambda_1)\bar\phi_1=0 $.  In particular we reassure that $G_{1,a}(V)=  \lambda_1(V)$ is a concave  functional. In general, higher eigenvalues $\lambda_j=\lambda_j(V)$ are {\em not} concave functions if $j>1$.  However, if  $\vec{\beta}:= (\beta_1, \ldots \beta_k)$ satisfies (\ref{betadef}) 
then we claim that $V\mapsto \sum_{j=1}^k \beta_j \lambda_j(V)$ is concave. Indeed:
\begin{lemma}\label{lemma2.2}
\be\label{Gdef} G_{\beta,a}(V)=\sum_{j=1}^k \beta_j \lambda_j(V)\  \ee
		where $\lambda_j(V)$ are the $k$ lowest  eigenvalues of the operator $L^W-aV$ arranged by order
		$$ \lambda_1(V)<\lambda_2(V)\leq \ldots \leq \lambda_k(V) \ . $$
	Moreover, the minimum  in (\ref{minlin}) is obtained at the  eigenfunction  $\bar\phi_j$ of  $L^W-aV$ corresponding to $\lambda_j$.

	\end{lemma}
 Recall the definition of sup-gradient of a concave  functional $G$ on a vector space $C^\infty_0(\R^n)$ at $V$:
$$ \partial_V G:= \left\{ \zeta\in (C_0^\infty)^{'}; \ \ G(Z)\leq G(V)+ \left<Z-V, \zeta\right> \ \forall  Z\in C_0^\infty\right\}$$
while $G$ is differentiable at $V$ if $\partial_V G$ is a singleton. 
\begin{cor}\label{cor2.1}
	The sup-gradient of the functional  $G_{\beta, a}$ on $C_0^\infty(\R^n)$ is contained in $\Ll^1$. In fact 
	$\partial_VG_{\beta, a}=a \sum_{j=1}^k \beta_j|\bar\phi_j|^2$ where $\bar\phi_j\in\Ll^2$ is a  normalized  eigenstate of $L^W-aV$ corresponding to the $j-$ eigenvalue.  So, in particular, $\|\partial_VG_{\beta, a}\|_1=a$. If all eigenvalues of $L^W-aV$ are simple then $G_{\beta,a}$ is differentiable at $V$.
\end{cor}
	\begin{proof}
		
		Since $V\mapsto \sum_{j=1}^k \beta_j \left<(L^W-aV)\phi_j, \phi_j\right>$  is a linear functional, (\ref{minlin}) would imply, in particular, that the functional $G_{\beta,a}$ is, indeed,  a concave  one.
		
		Let $\bar\phi_j$ be the normalized eigenvalues of $L^W-aV$ corresponding to $\lambda_j(V)$. 
		Fix some $m\geq j$ and let $\bH_m=  Sp(\bar\phi_1, \ldots \bar\phi_m)$. Let us restrict the supremum  (\ref{minlin}) to $\bH^k_m:=\{ \vec\phi:= (\phi_1, \ldots \phi_k), \phi_j\in\bH_m\}\subset \bH^k$.   
		
		Then
		$$ \phi_j= \sum_{i=1}^m\left<\phi_j, \bar\phi_i\right>\bar\phi_i , \ \ \ (L^W-aV)\phi_j= \sum_{i=1}^m\lambda_i\left<\phi_j, \bar\phi_i\right>\bar\phi_i \ \ . $$
	Define $\beta_{k+1}= \ldots =\beta_m=0$. Then we can write, for any $\vec\phi\in \bH^k_m$ 
		\be\label{beta=gamma}\sum_{j=1}^k \beta_j \left<(L^W-aV)\phi_j, \phi_j\right> =\sum_{i=1}^m\sum_{j=1}^m \beta_j\lambda_i|\left<\phi_j, \bar\phi_i\right>|^2 \ . \ee
		Denote now $\gamma_{i,j}:= |\left<\phi_j, \bar\phi_i\right>|^2$. Then $\{\gamma_{i,j}\}$ is $m\times m$, bi-stochastic matrix, i.e $\sum_{i=1}^m \gamma_{i,j}=\sum_{j=1}^m\gamma_{i,j}=1$ for all $i,j=1\ldots m$.  Consider now the infimum  of $\sum_{i=1}^m\sum_{j=1}^m \tilde\gamma_{i,j}\lambda_i\beta_j$ over all bi-stochastic martices $\{\tilde\gamma_{i,j}\}$.  By Krain-Milman theorem, the minimum is obtained on an extreme point in the convex set of bi- stochastic matrices. By Birkhoff theorem, the extreme points are permutations so, from(\ref{beta=gamma}) 
		$$ \forall \vec\phi\in \bH^k_m, \ \ \  \sum_{j=1}^k \beta_j \left<(L^W-aV)\phi_j, \phi_j\right>\geq \sum_{j=1}^m \beta_{\pi(j)} \lambda_j$$
		for some permutation $\pi:\{1, \ldots m\}\mapsto\{1, \ldots m\}$.  Now, recall that $\beta_j$ are assumed to be in descending order while $\lambda_j$ are in ascending  order by definition. By the discrete rearangment theorem of Hardy, Littelwood and Polya \cite{Polya} we obtain that the maximum on the right above is obrained at the identity permutation $\pi(i)=i$,  that is, at the identity matrix $\tilde\gamma_{i,j}:=\left<\phi_j,\bar\phi_i\right>=\delta_{i,j} $.  This implies thet the eigenbasis $\bar\phi_1, \ldots \bar\phi_k$ of the $k$ leading eigenvalues is the minimizer of (\ref{minlin}) on $\bH_m^k$ for any $m\geq k$.

	In particular,  the minimizer of (\ref{minlin}) in $\bH_m^k$ is independent of $m$, as long as $m\geq k$.  Suppose there exists some $\vec\psi\in \bH^k$ which is not contained in and finite dimensional subspace generated by eigenstates, for which (\ref{minlin}) is strictly smaller than its value on the first $k-$ leading eigenspace. Since the eigenstates of the Schrödinger operator under assumption (\ref{Winf}) generate the whole space we can find, for a sufficiently large $m$,  an orthonormal base in $\bH_m^k$ for on which the left side of (\ref{minlin}) is strictly larger than $\sum_{j=1}^k\beta_j\lambda_j(V)$, and we get a contradiction for this value of $m$. 
	
	\end{proof}
From Corollary \ref{cor2.1} and (\ref{ELalpha}, \ref{IinvD})  
	It follows that the Euler-Lagrange equation corresponding to ${\cal H}^{W,\alpha}_{\beta,a}$ is 
	$$ (-\Delta)^{\alpha/2} V - \sum_{j=1}^k\beta_j|\phi_j|^2=0   \ \ \ \Longleftrightarrow   V=I_\alpha*(
\sum_{j=1}^k\beta_j|\phi_j|^2	)$$
	where $\phi_j$  are the normalized eigenfunction corresponding to $\lambda_j( V)$. In particular we obtain
the proof of Theorem \ref{maintheorem}-(i): 
\begin{cor}\label{cor2.3}
	If $\bar V$ is a minimizer of ${\cal H}^{W,\alpha}_{\beta,a}$ then $\bar V=\sum_{j=1}^k\beta_j I_\alpha*|\bar\phi_j|^2$ where $\bar\phi_j$ are the normalized eigenfunction corresponding to $\lambda_j(\bar V)$. In particular, $\{\lambda_j(\bar V), \bar\phi_j\}$  is a solution of the Choquard system (\ref{firsteq}, \ref{rie}). 
	\end{cor}

\begin{lemma}\label{lemma2.3}
	Suppose $\alpha\in (0,2]$, $3\leq n\leq 4+\alpha$ and  $V\in C_0^\infty(\R^n)$  is bounded in $\dot\bH^{\alpha/2}$. If $\phi$ is a normalized eigenfunction of $a^{-1}L^W-V$ then $\|\nabla \phi\|_2$ , $\int W|\phi|^2dx$ 
	 and $\|\phi\|_{2n/(n-2)}$ are  bounded in terms of $\|| V\||_{\alpha/2}$ and the corresponding eigenvalue $\lambda$.   
	
\end{lemma}
\begin{proof}
	By assumption, $\|\phi\|_2=1$ and satisfy
	$$ (-\Delta \phi+W)\phi -aV\phi-\lambda \phi =0\ . $$
	Multiply by $\phi$ and integrate to obtain
	\be\label{eigenfunc} \|\nabla \phi\|_2^2 -a\int V|\phi|^2 dx+\int W|\phi|^2dx - \lambda=0 \ . \ee

	By the critical Sobolev inequality (Lemma \ref{critsob}) and and Holder inequality
	\be\label{cSH} \int V|\phi|^2dx \leq \|V\|_{\frac{2n}{n-\alpha}}\|\phi|^2\|_{\frac{2n}{n+\alpha}}\leq S_{n,\alpha} \||V\||_{\alpha/2}^2\|\phi|^2\|_{\frac{2n}{n+\alpha}}\ee
 By the Gagliardo-Nirenberg interpolation inequality
	 \cite{NL}
	 $$ \|\phi\|_p \leq C(\theta)\|\nabla \phi\|_2^\theta \|\phi\|_2^{1-\theta}$$
	 	where $C(\theta)$ is independent of $\phi$, $p=2n/(n-2\theta)$ whenever $\theta\in[0,1]$.  Since $\|\phi\|_2=1$ we get
	 	\be\label{ugradu} \||\phi|^2\|_{p/2}= \|\phi\|^2_p\leq C^2(\theta)\| \nabla \phi\|_2^{2\theta}\ . \ee
	 	Let now $p/2=\frac{2n}{n+\alpha}$, corresponding to $\theta=(n-\alpha)/4$, we obtain from (\ref{cSH})
	 	$$  \int V|\phi|^2dx \leq S_{n,\alpha} C^2(\theta) \||V\||_{\alpha/2}^2 \|\nabla \phi\|_2^{2\theta} \ $$
	 	where $\theta <1$ if $3\leq n< 4+\alpha$. Substitute it in (\ref{eigenfunc}) we obtain the upper estimate on $\|\nabla \phi\|_2$ and $\int W|\phi|^2dx$.  Finally setting $\theta=1$ corresponding to $p=2n/(n-2)$ we obtain from (\ref{ugradu}) the estimate on $\|\phi\|_{2n/(n-2)}$. 
 \end{proof}
\begin{lemma}\label{lemma3.2} If $2<n<4+\alpha$, $0<\alpha\leq 2$ then 
	$C_0^\infty(\R^n)\ni V\mapsto \lambda_j^{(V)}$ is  continuous on  bounded sets in $\dot\bH_{\alpha/2}(\R^n)$    with respect to Lebesgue norms  $\Ll^{q}(\R^n)$, where $n/2\leq q\leq \infty$. In particular $V\mapsto\lambda_j^{(V)}$ can be extended as a continuous function on $\dot\bH_{\alpha/2}$.
	
\end{lemma}
\begin{proof}
	By Lemma \ref{lemma2.2}, there exists $\vec{\phi}^{(V)}\in \bH^k_1$ such that  
	$$ G_{\beta, a}(V) =  \inf_{\vec{\phi}\in \oplus^k\bH^1}\sum_{j=1}^k \beta_j \left<(L^W-aV)\phi_j, \phi_j\right> \ . $$
	Thus, for $\tilde V_1, \tilde V_2$ bounded in $\dot\bH^{\alpha/2}$,
	$$ G_{\beta,a}(\tilde V_1)-G_{\beta,a}(\tilde V_2)\leq  \sum_{j=1}^k \beta_j \left<(L^W-a\tilde V_1)\phi^{(\tilde V_2)}_j, \phi^{(\tilde V_2)}_j\right> 
	- \sum_{j=1}^k \beta_j \left<(L^W-a\tilde V_2)\phi^{(\tilde V_2)}_j, \phi^{(\tilde V_2)}_j\right> $$ 
	$$ = a  \sum_{j=1}^k \beta_j \left<(\tilde V_2-\tilde V_1)\phi^{(\tilde V_2)}_j, \phi^{(\tilde V^2)}_j\right>\leq a  \sum_{j=1}^k \beta_j \|\tilde V_2-\tilde V_1\|_q\||\phi^{(\tilde V_2)}_j|^2\|_{\frac{q}{q-1}} \ . $$
	By Lemma \ref{lemma2.3}, $\||\phi_j^{(V)}|^2\|_p$ is bounded in terms of the norm of $\||V\||_{\alpha/2}$ for $1\leq p\leq n/(n-2)$. It follows that 
$G_{\beta,a}(\tilde V_1)-G_{\beta,a}(\tilde V_2)$ is bounded in terms of $\|\tilde V_2-\tilde V_1\|_{q}$ for $n/2\leq q\leq\infty$, so $G_{\beta,a}$ is continuous in these norms. Since $n/2< \frac{2n}{n-\alpha}$ by assumption and $\dot\bH_{\alpha/2}(\R^n)$ is embedded in $\Ll^{2n/(n-\alpha)}(\R^n)$, we obtain the continuous extension of $G_{\beta,a}$ on $\dot\bH_{\alpha/2}(\R^n)$ . 
	
	Finally, to continuity of each eigenvalue $\lambda_j$ is obtained by subtraction $G_{(\beta_1, \ldots \beta_j), a}(V)-G_{(\beta_1, \ldots \beta_{j-1}), a}(V)\equiv \beta_j\lambda_j^{(V)}$ by Lemma \ref{lemma2.2}.

\end{proof}
\subsection{Lower limit of the dual functional}
Recall the  Lieb-Thirring   inequality  for Schrodinger operator : 
\begin{theorem}\cite{LT}
	For the Schrödinger operator $-\Delta +V$ on $\R^n$ with a real valued potential $V$ 
the numbers 
	$\mu _{1}(V)\leq \mu _{2}(V)\leq \dots \leq 0$  denote the (not necessarily finite) sequence of its negative eigenvalues. Then, for $n\geq 3$ and $\gamma\geq 0$
\be\label{LTu}	\sum_{j; \mu_j(V)<0} |\mu_j(V)|^\gamma\leq L_{\gamma,n}\int V_-^{n/2+\gamma} dx \ \ee
where $V_-= \max\{0, -V\}$ and $L_{\gamma,n}$ is independent of $V$. 

\end{theorem}

\begin{prop}\label{prop2.1}
	The functional $V \mapsto \frac{a}{2}\left<V,V\right>_{\alpha/2} + G_{\beta, a}(V)$ is bounded from below  on $\dot\bH_{\alpha/2}$ for any $a>0$ if $3\leq n<2+\alpha$.    if $n=3, \alpha=1$  or $n=4, \alpha=2$ there exists $a=a^{(n)}_c(\vec{\beta})>0$ independent of $W$ for which the functional is bounded from below if $a< a^{(n)}_c(\vec{\beta})$ and unbounded if $a>a^{(n)}_c(\vec{\beta})$. 
	
	Moreover, in the cases $n=3$ and $n=4, a<a^{(n)}_c(\vec{\beta})$  the functional is coersive on $\dot\bH_{\alpha/2}(\R^n)$, namely
		\be\label{coersive} \lim_{\||V\||_{\alpha/2}\rightarrow\infty} \frac{1}{2}\left< V,V\right>_{\alpha/2} + G_{\beta, a}(V)=\infty \ . \ee
\end{prop}
\begin{proof}
	Recall that $\lambda_j(V)$ are the eigenvalues of $L^W -aV=-\Delta+ W-aV$.  Since $W\geq 0$  it follows that $\lambda_j(V)\geq\mu_j(aV)$.  Hence  $G_{\beta, a}(V):= \sum_{j=1}^k\beta_j\lambda_j(V)\geq  -\sum_{j; \mu_j(aV)<0}\beta_j |\mu_j(aV)|$. 
	By Holder inequality, for $\gamma\geq 1$,  $\gamma^{'}=\gamma/(\gamma-1)$ and (\ref{LTu})
	$$G_{\beta,a}(V) \geq -\left(
	\sum_{j=1}^k	|\beta_j|^{\gamma^{'}}\right)^{1/\gamma^{'}}
	\left( \sum_{j; \mu_j(aV)<0}| \mu_j(aV)|^\gamma\right)^{1/\gamma} \geq$$
	$$ -a^{1+n/2\gamma}L_{\gamma,n}^{1/\gamma}
	\|\vec{\beta}\|_{\gamma^{'}}\left(\int V_-^{n/2+\gamma} dx\right)^{1/\gamma}
	\ . $$
	Set now $\gamma=\frac{2n}{n-\alpha}-n/2\equiv  \frac{(4+\alpha)n-n^2}{2(n-\alpha)}$. 
	Then, if $2<n<2+\sqrt{1+3\alpha}$ we get $\gamma^{'}_{n,\alpha}\geq 1$ and
	$$ G_{\beta,a}(V)\geq -
	a^{\frac{4}{4+\alpha-n}} L_{\gamma,n}^{1/\gamma}
	\|\beta\|_{\gamma^{'}_{n,\alpha} }
	\left(\int_{\R^n} V_+^{\frac{2n}{n-\alpha}}\right)
	^{\frac{2(n-\alpha)}{(4+\alpha)n-n^2}} \ . $$
	Using the critical Sobolev inequality
	$$ G_{\beta,a}(V)\geq -
	a^{\frac{4}{4+\alpha-n}} L_{\gamma,n}^{1/\gamma}
	\|\beta\|_{\gamma^{'}_{n,\alpha} }S_{n,\alpha}^{\frac{4}{(4+\alpha)-n}}\left< V,V\right>_{\alpha/2} ^{\frac{2}{(4+\alpha)-n}}$$
	hence 
	$$ \frac{a}{2}
	\left< V,V\right>_{\alpha/2}  +G_{\beta,a}(V) \geq  
	$$
	\be\label{VVbeta}a\left< V,V\right>_{\alpha/2} ^{\frac{2}{4+\alpha-n}}\left( \frac{1}{2}
\left< V,V\right>_{\alpha/2} ^{1-\frac{2}{4+\alpha-n}}   -a^{\frac{n-\alpha}{4+\alpha-n}} 
	L_{\gamma,n}^{1/\gamma}\|\beta\|_{\gamma^{'}_{n,\alpha} }S_{n,\alpha}^{\frac{4}{(4+\alpha)-n}}\right)
	\ee
It follows that ${\cal H}^{W,\alpha}_{\beta,a}$ is coersive for any $a>0$ if  $3\leq n<2+\alpha$. If $n=2+\alpha$ then the functional is coersive if 
$a<\frac{S_{n,\alpha}^{\frac{4}{n-(4+\alpha)}}}{2L_{\gamma,n}^{1/\gamma}}|\beta|_{\gamma^{'}_{n,\alpha}}^{-1}$. Note that $\gamma^{'}_{n,\alpha}=\infty$ for  $n=3, \alpha=1$ and 
$\gamma^{'}_{n,\alpha}=2$ for $n=4, \alpha=2$. Hence coersivity holds if
\begin{itemize}
	\item $(\alpha,n)=(1,3)$: \ \ \ $a<\frac{S_{3,1}^{-2}}{2L_{\gamma,3}^{1/\gamma}}
	|\beta|_{\infty}^{-1}$
		\item $(\alpha,n)=(2,4)$: \ \ \ $a<\frac{S_{4,12}^{-2}}{2L_{\gamma,4}^{1/\gamma}}
	|\beta|_{2}^{-1}$
\end{itemize}

		We now prove the existence of a critical strengh $a_c(\beta)$ in both cases. 
		For a given, non-negative  $V\in \dot\bH^{\alpha/2}$, let $k(V)$ be the number of negative  eigenvalues of $ -\Delta - aV$, enumerated by order  $\lambda_1^0(V)< \lambda_2^0(V)\leq \ldots \lambda_{k(V)}^0(V) < 0$. Denote 
		$G^0_{\beta,a}(V)=\sum_{1}^{k\wedge k(V)}\beta_j\lambda_j^0(V)$.  Let $\bar\phi_j^0$ be the corresponding eigenfunctions of $-\Delta-aV$.  From the variational characterization of $G_{\beta,a}$ introduced in Lemma \ref{lemma2.2} we may obtain
		\be\label{G0g} G^0_{\beta,a} (V)\leq G_{\beta,a} (V)\leq G^0_{\beta,a}(V) + \sum_{j=1}^{k\wedge k(V)}\beta_j\int W|\phi_j^0|^2 +O(1)\ee
		where $O(1)$ stands for some constant independent of $V$. \footnote{ We may estimate this constant by $\sum_{k\wedge k(V)+1}^k \beta_j \lambda^w_j$ where $\lambda_j^w$ is the $j-$ eigenvalue of $H_0= -\Delta +W$.}
		
			Substitute now  now  $ V/\sqrt{a}$ for $V$. Then $\frac{a}{2}\left< V/\sqrt{a},V/\sqrt{a}\right>_{\alpha/2} +G^0_{\beta, a}(V/\sqrt{a}) = \frac{1}{2}\left< V,V\right>_{\alpha/2} +G^0_{\beta, a}(V/\sqrt{a}) $. By definition $G^0_{\beta,a}(V/\sqrt{a})=\sum_{1}^{k\wedge k(V/\sqrt{a})}\beta_j\lambda_j^0(V/\sqrt{a})$, while 
		$\lambda_j^0(V/\sqrt{a})$ is a negative  eigenvalue of $-\Delta +W +\sqrt{a}V$. Thus, if $V<0$ somewhere then  $\lim_{a\rightarrow\infty}G^0_{\beta,a} (V/\sqrt{a}) =-\infty$.  In particular it follows that 
		\be\label{ifthen} \text{if}\ \  a>0 \ \ \text{large enough}  \ \text{then} \ \exists  V\in \dot\bH_{\alpha/2}  \ \text{for which} \ \frac{a}{2}\left< V,V\right>_{\alpha/2}+G^0_{\beta, a}(V)<0\ee 
		
		Now apply the transformation $V\mapsto V_\delta(x):= \delta^{2}V(\delta x)$, where $\delta>0$. 
	We obtain that $\lambda^0_j(V_\delta) = \delta^2\lambda^0_j(V)$ (in particular, $k(V_\delta)=k(V)$), while   $\phi_j^{0,\delta} := \delta^{n/2}\phi_j^0(\delta x)$ is the corresponding normalized  eigenfunction. Hence
	$G^0_{\beta,a}(V_\delta)=\delta^2G^0_{\beta,a}(V)$ so, by (\ref{G0g}), $G_{\beta,a}(V_\delta)\leq \delta^2 G^0_{\beta,a}(V)+\sum_{j=1}^{k\wedge k(V)}\beta_j\int W|\phi_{j, \delta}^0|^2 +O(1)$.

	Next, we obtain for both $n=3, \alpha=1$ and $n=4, \alpha=2$ cases that the quadratic form scale the same:  $\left< V_\delta,V_\delta\right>_{\alpha/2}= \delta^2\left< V,V\right>_{\alpha/2}$ so  
	\be\label{G<G0} \frac{a}{2}\left< V_\delta,V_\delta\right>_{\alpha/2}+G_{\beta, a}(V_\delta) \leq  \delta^2\left[\frac{a}{2}\left< V,V\right>_{\alpha/2} +G^0_{\beta, a}(V) \right] + \sum_{j=1}^{k\wedge k(V)}\beta_j\int W|\phi_{j, \delta}^0|^2 +O(1)\ . \ee
 By (\ref{Winf}) we also get $\lim_{\delta\rightarrow\infty} \int W|\phi^0_{j,\delta}|^2 =W(0)=0$, so, using (\ref{ifthen}) we obtain the existence of $V$ for which 
	$\frac{a}{2}\left< V_\delta,V_\delta\right>_{\alpha/2} +G_{\beta, a}(V_\delta)\rightarrow -\infty$ as $\delta\rightarrow\infty$, if $a>0$ is large enough.

Now let 
	$$ a_c(\vec{\beta}) =\inf\left\{a>0;\ \ \inf_{V\in \dot\bH^{\alpha/2}}\frac{a}{2}\left<V,V\right>_{\alpha/2} +G^0_{\beta, a}(V)<0 \right\}   \ . $$
It follows that  $\infty> a_c(\beta)>0$ and is independent of $W$  for any $\vec\beta$ in the cases $n=3, \alpha=1$ and $n=4, \alpha=2$. 
\end{proof}

\subsection{Existence of minimizers of the local problem}
When attempting to prove the existence of minimizers to the functional  ${\cal H}^{W,\alpha}_{\beta,a}$ (\ref{Hdef}) we face the problem of compactness of the space $\dot\bH_{\alpha/2}$.  So, we start by considering the subspace  of $\dot\bH_{\alpha/2}(B^n_R)\subset \dot\bH_{\alpha/2}\R^n)$, obtained by the closure of $C_0^\infty$ of functions {\em supported on the ball } $B_R^n:=\{x\in \R^n; \ |x|<R\}$ under the induced $\||\cdot\||_{\alpha/2}$ norm (section \ref{crush}). 

Note that, by this definition,  $V\in \dot\bH_{\alpha/2}(B^n_R)$ is defined {\em over the whole} of $\R^n$, and is identically zero on $\R^n-B^n_R$. 

By 

\begin{lemma}\label{lemma3.1}
	Given $\vec\beta$ satisfying (\ref{betadef}), $R>0$, $3\leq n<2+\alpha$ or either $n=3, \alpha=1$  or $n=4, \alpha=2$. Then 
  there exists a minimizer $\bar{V}_R$ of ${\cal H}^{W,\alpha}_{\beta,a}$  restricted to  $\dot\bH_{\alpha/2}(B_R^n)$. Moreover, 
  \be\label{Ubr=}I_\alpha^{B_R^n}(\bar V_R)= \left(\sum_{j=1}^k \beta_j|\bar\phi^R_j|^2\right)\ee
	where $\bar\phi_j^R$ are the normalized  eigenstates of 
	$L^W- a\bar V_R$ in $\R^n$. 
\end{lemma}
\begin{proof}
Let $V_n\subset \dot\bH_{\alpha/2}(B^n_R)$ be a minimizing sequence of ${\cal H}^{W,\alpha}_{\beta,a}$ . Since ${\cal H}^{W,\alpha}_{\beta,a}$  is bounded from below by Proposition  \ref{prop2.1} we get
$$ \lim_{n\rightarrow\infty} {\cal H}^{W,\alpha}_{\beta,a}(V_n)= \inf_{V\in \dot\bH_{\alpha/2}(B^n_R)} {\cal H}^{W,\alpha}_{\beta,a}(V) \ . $$
Since $\dot\bH_{\alpha/2}(B^n_R)$ is weakly compact  and the functional is coersive (\ref{coersive}) there exists a weak limit $\bar V_R\in \dot\bH_{\alpha/2}(B^n_R)$ of this sequence. Moreover, by Sobolev compact embedding, $V_n$ converges strongly to $\bar V_R$ in $\Ll^q(B_B^n)$ for any $1\leq q<2n/(n-\alpha)$.  Since $n<2+\alpha$ then $n/2< 2n/(n-\alpha)$ and 
by Lemma \ref{lemma3.2}
\be\label{wealls}\lim_{n\rightarrow\infty} G_{\beta,a}(V_n)= G_{\beta,a}(\bar V_R) \ . \ee

Since $V \mapsto\|| V\||_{\alpha/2}^2$ is l.s.c , it follows 
that
$$\lim_{n\rightarrow \infty} \left<V_n, V_n\right>_{\alpha/2}dx \geq  \left<\bar V_R, \bar V_R\right>_{\alpha/2}\ . $$

This and (\ref{wealls}) imply that $\bar V_R$ is, indeed, a minimizer of  ${\cal H}^{W,\alpha}_{\beta,a}$ on $\dot\bH_{\alpha/2}(B^n_R)$. 

Finally, (\ref{Ubr=}) follows from (\ref{defIomega}) while taking $\Omega=B_R^n$. 
 \end{proof}

\subsection{Proof of Theorem  \ref{maintheorem}(ii, iii)}
	Let $V_m$ be a minimizing sequence for ${\cal H}^{W,\alpha}_{\beta,a}$ in $\dot\bH_{\alpha/2}(\R^n)$. Since $C^\infty_0(\R^n)$ is dense in $\dot\bH_{\alpha/2}(\R^n)$ by definition, we can assume that there exists a sequence $R_m\rightarrow\infty$ such that $V_m$ is supported in $B^n_{R_m}$. 

Let $\bar V_m$ be the  minimizers of ${\cal H}^{W,\alpha}_{\beta,a}$ 
on $\dot\bH_{\alpha/2}(B^n_R)$.  

Since ${\cal H}^{W,\alpha}_{\beta,a}(\bar V_m) \leq {\cal H}^{W,\alpha}_{\beta,a}(V_m)$ then $\bar V_m$ is a minimizing sequence of ${\cal H}^{W,\alpha}_{\beta,a}$ on $\dot\bH_{\alpha/2}(\R^n)$ as well.   Now, under the conditions of the Theorem we know by Proposition \ref{prop2.1} that ${\cal H}^{W,\alpha}_{\beta,a}$ is bounded from below on $\dot\bH_{\alpha/2}(\R^n)$ and coersive (\ref{coersive}), so $\||\bar V_m\||_{\alpha/2}$ are uniformly bounded. 
 Let $\bar\phi_j^m$ be the normalized eigenfunctions of $L^W-a\bar V_m$. 
By Lemma \ref{lemma2.3} we obtain that $\|\nabla\bar\phi^m_j\|_2$ and  $\int_{\R^n} W|\bar\phi_j^m|^2$ 
and $\||\bar\phi_j^R|^2\|_{n/(n-2)}$ 
are uniformly bounded on $\R^n$.
 In addition,  $\|\bar\phi_j^m\|_2=1$ by definition.
 In particular, $\bar\phi_j^m$ are in the space $\bH^1$ (c.f Definition  \ref{def1.1}) . Using the first part of Lemma \ref{lemma2.1} 
  we obtain a subsequence (denoted by the index $m$ ) along which $\bar\rho_m:= \sum_{j=1}^k \beta_j |\bar\phi_j^m|^2$ converges in $\Ll^p(\R^n)$ for any $ p < n/(n-2)\equiv 2^*/2$, while $\bar V_m=I^{B^n_R}_\alpha(\bar\rho_m)$. Since $n/2\leq  n/(n-2)$ for $n=3,4$, Lemma \ref{lemapqrietz} implies the convergence of $\bar V_m$ to $\bar V$ in $\Ll^q(\R^n)$ for any $1\leq q<\infty$.  By lower semi continuity we obtain that $\bar V\in \dot\bH_{\alpha/2}(\R^n)$ and $\left<\bar V, \bar V\right>_{\alpha/2}\leq \lim_{m\rightarrow\infty}\left<\bar V_m,\bar V_m\right>_{\alpha/2}$.
  
    In addition, Lemma \ref{lemma3.2} implies  that $G_{\beta,a}(\bar V_m)$ converges to $G_{\beta,a}(\bar V)$. 
This implies 
$$ {\cal H}^{W,\alpha}_{\beta,a}(\bar V) \leq \inf_{V\in \dot\bH_{\alpha/2}(\R^n)} {\cal H}^{W,\alpha}_{\beta,a}(V) $$
so $\bar V\in \dot\bH_{\alpha/2}(\R^n)$  is, indeed, a minimizer.  The proof of Theorem \ref{maintheorem} follows now from Lemma \ref{lema1.3} and Corollary \ref{cor2.3}.

\section{Further remarks}\label{FR}
It is interesting to consider the dependence of the solution to the Choquard system  on the probability vector $\vec{\beta}$. In  particular, the relation between the critical interaction strength $a(\beta)$ at dimension $n-4$ and  the universal critical value $\bar{a}_c$  corresponding to the scalar case $k=1$ (see (\ref{barphi})). 
\begin{description}
	\item[a)] {\em Estimate on $\bar a_c$}: In \cite{conmin} the critical value in case $\alpha=n-2$ is implicitly given as the $\Ll^2$ norm of the solution of equation (\ref{barphi}). However,  these solutions are not known explicitly. Here we introduce an estimate based on Hardy inequality
	$$ \int_{\R^n}|\nabla f|^2 \geq \left(\frac{n-2}{2}\right)^2\int_{\R^n} \frac{|f|^2}{|x|^2} \ $$
	for any $f\in C_0^\infty(\R^n)$. In particular it implies that the operator $-\Delta -V$ is non-negative in $\R^n$ for any $V\leq \left(\frac{n-2}{2}\right) ^2|x|^{-2}$. 
	
	As discussed in section \ref{backgr}, the functional $E^W_a$ is bounded from below on the unit ball of $\Ll^2$ iff $a\leq \bar a_c$. This implies, in particular, that if $A>\bar a_c$ there exists $\tilde\phi\in \bH^1$ for which 
		\be\label{E0} E^0_a(\tilde\phi):=\frac{1}{2} \int_{\R^n} |\nabla \tilde\phi|^2  - \frac{A}{4} \int_{\R^n} \left(I_{n-2}*|\tilde\phi|^2\right) |\tilde\phi|^2  <0 \ .  \ee
	Moreover, by Riesz's rearrangement theorem we can assume that this $\tilde\phi$ is radially symmetric. 
	
		In particular, for any $V\geq I_{n-2}*|\tilde\phi|^2$  
		\be\label{hiphi}-\Delta-( A/2) V\not\geq  0 \ . \ee 
		
		In the special case $n=4$, $I_2=(-\Delta)^{-1}$ is the fundamental solution  of the Laplacian.  Let $\rho:=|\tilde\phi|^2$ be this radial function. Then $U:= I_2*\rho$ is a solution of $\Delta U + \rho=0$. Thus
		\be\label{uar} r^{-3}\left( r^3 U^{'}\right)^{'} + \rho(r)=0 \ . \ee
		Let $m(r)=2\pi^2\int_0^r s^3\rho(s)ds$. In particular, $m(\cdot)$ is non-decreasing on $\R_+$, $m(0)=0$, and $m(r)\leq 1$ by assumption. Integrating (\ref{uar}) we get
		$$ r^3U^{'}(r)= -(2\pi^2)^{-1} m(r)  \Longrightarrow
		U(r)= (2\pi)^2\int_r^\infty \frac{m(s)}{s^3}ds\leq 2\pi^2 r^{-2} . $$ Thus, taking $V=2\pi^2r^{-2}$ in (\ref{hiphi}) we obtain a violation of the  Hardy inequality if $\pi^2 A$ is below the Hardy constant.   Since the Hardy constant $\left(\frac{n-2}{2}\right)^2=1$ for $n=4$ we get $A> \pi^{-2}$ for any $A>\bar a_c$, that is 
		$$ \bar a_c \geq \frac{1}{\pi^2}$$
		if $n=4$.

		It is not clear, at this point, if the above estimate holds for general dimension, since $I_{n-2}=(-\Delta)^{-1}$ only if $n=4$. There is, indeed, an estimate of the form
		$$ |x; I_\alpha* \rho(x)|>t| \leq c\left( \frac{c}{t}\|\rho\|_2\right)^{n/(n-\alpha)} $$
		(c.f \cite{Rietz}, eq. (2.12)) which, if $\rho$ is radial, is equivalent to 
		$$I_\alpha*\rho(r)\leq c^{\alpha}\omega_n^{(n-\alpha)/n}\|\rho\|_1 r^{\alpha-n}\  $$ 
		where $\omega_n$ is the surface area of the unit sphere $\mathbb{S}^{n-1}$. This suggests a similar estimate for $\bar a_c$ in for general $n$ and $\alpha =n-2$ using Hardy inequality. However, there is now known estimate (as far as we know) for the constant $c$. 
			\item[b)] {\em Relation between $\bar a_c$ and $a_\beta$}: 
			The inequality 
			$a(\beta)\geq \bar a_c$
			can be easily obtained for the critical case for any $\alpha=n-2$, $n\geq 3$, and any $\vec{\beta}$ satisfying (\ref{normalbeta}). Indeed, using Definition \ref{def1.1} and the polar  inequality
			$$ \left<|\phi_j|^2, I_\alpha* |\phi_i|^2\right>\leq \frac{1}{2}\left[ \left<|\phi_j|^2, I_\alpha* |\phi_j|^2\right>
			+ \left<|\phi_i|^2, I_\alpha* |\phi_i|^2\right>\right]$$
			we obtain 
				$${\cal E}^{(\alpha)}_{\beta,a}(\vec\phi)\geq \frac{1}{2} \sum_{j=1}^k\beta_j\left[  \left<\left< \phi_j,\phi_j\right>\right>_W -\frac{a }{2}
			\left<|\phi_j|^2, I_\alpha* |\phi_j|^2\right>\right] = \sum_{j=1}^k \beta_j E^W_a(\phi_j) $$
			where $E^W_a$ as defined in (\ref{EWdef}). It follows that ${\cal E}^{(\alpha)}_{\beta,a}$ is bounded on $\oplus^k\bH^1$ if $E^W_a$ is bounded on $\bH^1$. Since $E^W_a$ is bounded from below iff $a\leq \bar a_c$ (\cite{conmin}), the inequality $a(\beta)\geq \bar a_c$ follows. 
			
			In the case $n=3, \alpha=1$ and $n=4, \alpha=2$ we can say more about $a_c(\beta)$. By definition, $a>a_c(\beta)$ iff ${\cal H}^{W,\alpha}_{\beta,a}$ is unbounded from below on $\dot\bH_{\alpha/2}$. Using  (\ref{VVbeta}) we obtain that 
			 $a_c(\beta)>O(
			 |\beta|_{\infty}^{-1})$     
			 for $n=3$ and 
		 $a_c(\beta)> O(|\beta|^{-1}_2)$ for  $n=4$. 
			
			For an interesting conclusion from the above estimate, let $\vec{\beta}$ be the uniform vector $\vec{\beta}={\bf 1}_k:= k^{-1}(1, \ldots 1)\in \R^k$. Then $|\vec{\beta}|_2=k^{-1/2}$ (resp. $|\vec{\beta}|_\infty=k^{-1}$) so
			$$n=4 \Rightarrow a_c({\bf 1}_k)\geq O(k^{1/2})    \ \ \ \text{resp.} \ n=3 \Rightarrow\ ( a_c({\bf 1}_k)\geq O(k)  $$ for large $k$. 
		\item{[c]} 
	The  following alternative definitions of $I_\alpha$ and $(-\Delta)^{\alpha/2}$ is known \cite{Rietz, Ten}:
	$$ I_\alpha =\frac{1}{\Gamma(\alpha)} \int_0^\infty t^{\alpha/2-1} e^{t\Delta} dt \ \ \ \ ;  \ \ \  (-\Delta)^{\alpha/2}= \frac{1}{\Gamma(-\alpha)} \int_0^\infty t^{-\alpha/2-1} \left(e^{t\Delta} -I\right)dt$$
	where $e^{t\Delta}$ is the heat kernel on $\R^n$:
	$$ e^{t\Delta}= (4\pi t)^{-n/2} e^{-\frac{|x|^2}{4t}} \ . 
	$$
	We may, at least formally, substitute the kernel $e^{\Delta_\Omega t}$  of  the killing, Dirichlet problem  for the heat flow in a domain $\Omega\subset \R^n$ in the above expression, and obtain (again, at least formally...) an explicit expressions for $I_\alpha^\Omega$ and $(-\Delta_\Omega)^{\alpha/2}$, introduced implicitly in 
		(\ref{defIomega}). Such a representation can provide some insight on the trace of $I^\Omega_\alpha$ for $\alpha<2$. 
	
	\end{description}

\end{document}